\documentclass{llncs}

\usepackage{times}
\usepackage{helvet}
\usepackage{courier}

\usepackage{a4,a4wide}
\usepackage{amsmath}
\usepackage{latexsym,theorem}
\usepackage{epic,eepic}

\newtheorem{te}{Theorem}
\newtheorem{pr}[te]{Proposition}

{\theorembodyfont{\upshape}
\newtheorem{df}[te]{Definition}
\newtheorem{ex}[te]{Example}
}

\newcommand{\dep}{\ll}

\newcommand{\sCON}{\mathit{CON}}
\newcommand{\la}{\leftarrow}

\newcommand{\sLit}{\mathit{Lit}}

\newcommand{\snot}{\mathit{not} \;}
\newcommand{\sObj}{\mathit{Obj}}
\newcommand{\sDef}{\mathit{Def}}
\newcommand{\shead}{\mathit{head}}
\newcommand{\sbody}{\mathit{body}}
\newcommand{\sbodym}[1]{\mathit{(body(#1))^-}}
\newcommand{\sbodyp}[1]{\mathit{(body(#1))^+}}
\newcommand{\sP}{\mathit{pos}}

\newcommand{\sAt}{\mathit{At}}

\newcommand{\sSM}{\mathit{AS}}

\newcommand{\sM} [1] {\mathcal A_#1}
\newcommand{\argStA} [1] {\mathcal A_#1 = \langle Y_#1 \hookleftarrow X_#1 \rangle}

\newcommand{\argStC} [1] {\mathcal A_#1 = \langle \{ \shead(r_#1) \} \hookleftarrow \sbodym{r_#1} ; \sbodyp{r_#1} \rangle}

\newcommand{\argStSimple}[2] { \langle #1 \hookleftarrow #2 \rangle}
\newcommand{\argStFull}[3] {\langle #1 \hookleftarrow #2 ; #3 \rangle}
\newcommand{\argStName}[1]{\mathcal #1}

\newcommand{\sset}[1]{ \{ #1 \} }
\newcommand{\sCn}[2]{\mathit{Cn}_{\dep_{#1}}(#2)}

\title{Warranted Derivations of Preferred Answer Sets}
\author
{J\'an \v{S}efr\'anek and Alexander \v{S}imko}
\institute{Comenius University, Bratislava, Slovakia;  [sefranek,simko]@fmph.uniba.sk}

\begin{document}

\maketitle

\begin{abstract}
We are aiming at a semantics of logic programs with preferences defined on rules, which always selects a preferred answer set, if there is a non-empty set of (standard) answer sets of the given program. 

It is shown in a seminal paper by Brewka and Eiter that the goal mentioned above is incompatible with their second principle and it is not satisfied in their semantics of prioritized logic programs. Similarly, also according to other established semantics, based on a prescriptive approach, there are programs with standard answer sets, but without preferred answer sets.
 
According to the standard prescriptive approach no rule can be fired before a more preferred rule, unless the more preferred rule is  blocked. This is a rather imperative approach, in its spirit.

In our approach, rules can be blocked by more preferred rules, but the rules which are not blocked are handled in a more declarative style, their execution does not depend on the given preference relation on the rules.

An argumentation framework (different from the Dung's framework) is proposed in this paper. Argumentation structures are derived from the rules of a given program. An attack relation on argumentation structures is defined, which is derived from attacks of more preferred rules against the less preferred rules. Preferred answer sets correspond to complete argumentation structures, which are not blocked by other complete argumentation structures.
\end{abstract}


\section{Introduction}

The meaning of a nonmonotonic theory is often characterized by a set of {\em (alternative) belief sets}. It is appropriate to select sometimes some of the belief sets as more {\em preferred}. 

We are focused in this paper on {\em extended logic programs} with a preference relation on rules, see, e.g., \cite{be,dst,sw,w}. Such programs are denoted by the term {\em prioritized} logic programs in this paper.

It is suitable to require that some preferred answer sets can be {\em selected} from a non-empty set of standard answer sets of a (prioritized) logic program.


Unfortunately, there are prioritized logic programs with standard answer sets, but without preferred answer sets
according to the semantics of \cite{be} and also of \cite{dst} or \cite{w}. This feature is a consequence of the {\em prescriptive} approach to preference handling \cite{class}. According to that approach no rule can be fired before a more preferred rule, unless the more preferred rule is blocked. This is a rather imperative approach, in its spirit.

An investigation of basic {\em principles} which should be satisfied by any system containing a preference relation on defeasible rules is of fundamental importance. This type of research has been initialized in the seminal paper \cite{be}. Two basic principles are accepted in the paper.

The second of the principles is violated, if a function is assumed, which always
selects a non-empty subset of preferred answer sets from a non-empty set of all standard answer sets of a prioritized logic program.


We believe that the possibility to select always a preferred answer set from a non-empty set of standard answer sets is of critical importance. This goal requires to accept a {\em descriptive} approach to preference handling. The approach is characterized in \cite{ds,class} as follows: The order in which rules are applied, reflects their ``desirability'' -- it is desirable to apply the most preferred rules. 

In our approach, rules can be {\em blocked} by more preferred rules, but the rules which are not blocked are handled in a more declarative style. If we express this in terms of desirability, it is desirable to apply all (applicable) rules, which are not blocked by a more preferred rule.
 The execution of non-blocked rules does not depend on their order. Dependencies of more preferred rules on less preferred rules do not prevent the execution of non-blocked rules. 

Our goal is: 
\begin{itemize}
\item to modify the Principles proposed by \cite{be} in such a way that they do not contradict a selection of a non-empty set of preferred answer sets from the underlying non-empty set of standard answer sets,
\item to introduce a notion of preferred answer sets that
 satisfies the above mentioned modification.
\end{itemize}


The proposed method is inspired by \cite{gs}. While there defeasible rules are treated as (defeasible) arguments, here (defeasible) assumptions, sets of default negations, are considered as arguments. Reasoning about preferences in a logic program is here understood as a kind of argumentation. Sets of default literals can be viewed as defeasible arguments, which may be contradicted by consequences of some applicable rules. The preference relation on rules is used in order to ignore the attacks of less preferred arguments against more preferred arguments. The core problem is to transfer the preference relation defined on rules to argumentation structures and, consequently, to answer sets.\footnote{Our intuitions connected to the notion of argumentation structure and also the used constructions are different from Dung's arguments or from arguments of \cite{gs,ca}. On the other hand, we plan an elaboration of presented constructions aiming at a theory, which generalizes Dung's abstract argumentation framework, TMS, constructions given, e.g., by \cite{gs} or \cite{ca}. Anyway, this paper does not present a contribution to argumentation theory.}

The basic argumentation structures correspond to the rules of a given program. Derivation rules, which enable derivation of argumentation structures from the basic ones are defined. 
That derivation leads from the basic argumentation structures (corresponding to the rules of a given program) to argumentation structures corresponding to the rules of an negative equivalent of the given program introduced in \cite{dt}. 

Attacks of more preferred rules against the less preferred rules are transferred via another set of derivation rules to the attacks between argumentation structures. 
Preferred answer sets are defined over that background. They correspond to complete and non-blocked (warranted) argumentation structures.


The contributions of this paper are summarized as follows. A modified set of principles for preferred answer set specification is proposed. A new type of argumentation framework is constructed, which enables a selection of preferred answer sets. There are basic arguments (argumentation structures) and attacks in the framework and also derived arguments and attacks. Rules for derivation of argumentation structures and rules for derivation of attacks of some argumentation structures against other argumentation structures are defined. Preferred answer sets are defined in terms of complete and non-blocked (warranted) argumentation structures.  Finally, we emphasize that each program with a non-empty set of answer sets has a preferred answer set. 

A preliminary version of the presented research has been published in \cite{nmr08}. The main differences between the preliminary and the current version are summarized in the Conclusions.\footnote{They are described in technical terms, assuming a familiarity with this paper. Most importantly, the notion of preferred answer set is changed in this paper.}
An extended version of this paper with proofs is accessible as \cite{tr}.

\section{Preliminaries} \label{dlp}

The language of extended logic programs is used in this paper.

Let $\sAt$ be a set of atoms. The set of {\em objective literals} is defined as $\mathit{Obj} = {\sAt} \cup
\{ \neg \; A : A \in {\sAt} \}$. If $L$ is an objective literal then the expression of the form $\snot L$
is called {\em default} literal. Notation: $\mathit Def = \{ \snot L \mid L \in \sObj \}$. The set of literals 
$\sLit$ is defined as $\sObj \cup \mathit Def$. 

A convention: $\neg \neg A = A$, where $A \in {\sAt}$.
If $X$ is a set of objective literals, then $\snot X = \{ \snot L \mid L \in X \}$.

A {\em rule} is each expression of the form $L \leftarrow L_{1}, \dots, L_{k}$,
where $k \geq 0$, $L \in \sObj$ and $L_{i} \in \sLit$. If $r$ is a rule of the form as above, then $L$ is denoted by $\mathit{head}(r)$
and $\{ L_{1}, \dots, L_{k} \}$ by $\mathit{body}(r)$.
If $R$ is a set of rules, then $\shead(R) = \{ \shead(r) \mid r \in R \}$ and $\sbody(R) = \{ \sbody(r) \mid r \in R \}$. 
A finite set of rules is called {\em extended logic program} (program hereafter).

The set of {\em conflicting literals} is defined as 
$\sCON = \{ (L_1, L_2) \mid \; L_1 = \snot L_2 \vee L_1 = \neg L_2 \}$.
A set of literals $S$ is {\em consistent} if $(S \times S) \cap \sCON = \emptyset$.
An {\em interpretation} is a consistent set of literals. 
A {\em total} interpretation is an interpretation $I$
such that for each objective literal $L$ either $L \in I$ or $\mathit{not} \; L \in I$. 
A literal $L$ is {\em satisfied} in an interpretation $I$ iff $L \in I$ (notation: $I \models L$). A set of literals $S$ is satisfied in $I$ iff
$S \subseteq I$ (notation: $I \models S$). A rule $r$ is satisfied in $I$ iff $I \models \shead(r)$ whenever $I \models \sbody(r)$.

If $S$ is a set of literals, then we denote $S \cap \sObj$ by $S^{+}$ and $S \cap \sDef$ by $S^{-}$. 
Symbols $\sbodym{r}$ and $\sbodyp{r}$ are used here in that sense (notice that the usual meaning of $\sbody^{-}(r)$
is different). If $X \subseteq \sDef$ then $\sP(X) = \{ L \in \sObj \mid \snot L \in X \}$. Hence, $\snot \sP(\sbodym{r}) = \sbodym{r}$.
If $r$ is a rule, then the rule $\shead(r) \la \sbodyp{r}$ is denoted by $r^{+}$. 

An answer set of a program can be defined as follows (only consistent answer sets are defined).

\label{stable}
A total interpretation $S$ is an {\em answer set} of a program $P$ iff 
$S^{+}$ is the least model\footnote{$P^{+}$ is treated as definite logic program, i.e., each objective literal of the form $\neg A$, where $A \in \sAt$, is considered as a new atom.} of the program $P^{+} = \{ r^{+} \mid S \models \sbodym{r} \}$.
Note that an answer set $S$ is usually represented by $S^{+}$ (this convention is sometimes used also in this paper). 

The set of all answer sets of a program $P$ is denoted by $\sSM(P)$. 
A program is called {\em coherent} iff it has an answer set.

Strict partial order is a binary relation, which is irreflexive, transitive and, consequently, asymmetric.

A {\em prioritized logic program} is usually defined as a triple $(P, \prec, \mathcal N)$, where $P$ is a program, $\prec$ a strict partial order on rules of $P$ and a function $\mathcal N$ assigns names to rules of $P$.
If $r_1 \prec r_2$ it is said that $r_2$ is more preferred than $r_1$.

We will not use $\mathcal N$. If a symbol $r$ is used in this paper in order to denote a rule, then $r$ is considered as the name of that rule (no different name $\mathcal N(r)$ is introduced).

\section{Argumentation Structures}
\label{sas}

Our aim is to transfer a preference relation defined on rules to a preference relation on answer sets and, finally, to a notion of preferred answer sets. To that end argumentation structures are introduced. The basic argumentation structures correspond to rules. Some more general types of argumentation structures are derived from the basic argumentation structures. A special type of argumentation structures corresponds to answer sets.

\begin{df}[$\dep_P$, \cite{js}]
\label{depP}
An objective literal $L$ \emph{depends} on a set of default literals $W \subseteq \sDef$ {\em with respect to a program} $P$ 
($L \dep_P W$\footnote{$L \dep_P W$ could be defined as $T_P^\omega(W)$ and $Cn_{\dep_P}(W)$ as $T_P^\omega(W)$}) iff 
there is a sequence of rules $\langle r_1, \dots, r_k \rangle$, $k \geq 1$, $r_i \in P$ such that
\begin{itemize}
\item $\shead(r_k) = L$,
\item $W \models \sbody(r_1)$,
\item for each $i$, $1 \leq i < k$, $W \cup \{\shead(r_1), \dots, \shead(r_i) \} \models \sbody(r_{i+1})$.
\end{itemize}

The set $\{L \in Lit \mid L \dep_P W \} \cup W$ is denoted by $\sCn{P}{W}$.

$W \subseteq \sDef$ is {\em self-consistent} w.r.t. a program $P$ iff $\sCn{P}{W}$ is consistent.
$\Box$
\end{df}

If $Z \subseteq \sObj$, we will use sometimes the notation $\sCn{P \cup Z}{W}$, assuming that the program $P$ is extended by the set of facts $Z$.

\begin{df}[Dependency structure]
\label{argument}
Let $P$ be a program.

A self-consistent set $X \subseteq \sDef$ is called an {\em argument} w.r.t. the program $P$ for a consistent set of objective literals $Y$, given a set of objective literals $Z$ iff
\begin{enumerate}
\item $\sP(X) \cap Z = \emptyset$,
\item $Y \subseteq \sCn{P \cup Z}{X}$.
\end{enumerate}
We will use the notation $\langle Y \hookleftarrow X; Z \rangle$ and the triple denoted by it is called a {\em dependency structure} (w.r.t. $P$).
$\Box$
\end{df}


If $Z = \emptyset$ also a shortened notation $\langle Y \hookleftarrow X \rangle$ can be used. We will omit sometimes the phrase ``w.r.t. $P$'' and speak simply about dependency structures and arguments, if the corresponding program is clear from the context.

Basic argumentation structures comply with Definition \ref{argument} of dependency structures, if some conditions are satisfied:

\begin{df}[Basic argumentation structure]
\label{basicAS}

Let $r \in P$ be a rule such that
\begin{itemize}
\item $\sbodym{r}$ is self-consistent and
\item $\sP(\sbodym{r}) \cap \sbodyp{r} = \emptyset$.
\end{itemize}
Then $\argStName{A} = \argStFull{ \{ \shead(r) \} }{\sbodym{r}}{\sbodyp{r}}$ is called a {\em basic} argumentation structure.
$\Box$
\end{df}
\begin{pr}
Each basic argumentation structure is a dependency structure. 
$\Box$
\end{pr}

We emphasize that only {\em self-consistent} arguments for {\em consistent} sets of objective literals are considered in this paper. Hence, programs as $P = \{ p \la \snot p \}$ or $Q = \{ p \la \snot q; \neg p \la \snot q \}$ are irrelevant for our constructions.

Some dependency structures can be derived from the basic argumentation structures.
Only the dependency structures derived from the basic argumentation structures using derivation rules from Definition \ref{derR} are of interest in the rest of this paper, we will use the term {\em argumentation structure} for dependency structures derived from basic argumentation structures using derivation rules.

Derivation rules are motivated in Example \ref{zaciatok}. 

\begin{df}[Derivation rules and argumentation structures]
\label{derR}
Each basic argumentation structure is an argumentation structure.
Let $P$ be a program.
\begin{description}
\item[R1]
Let $r_1, r_2 \in P$, $\argStName{A_1} = \argStFull{\sset{\shead(r_1)}}{X_1}{Z_1}$ and $\argStName{A_2} = \argStFull{\sset{\shead(r_2)}}{\sbodym{r_2}}{\sbodyp{r_2}}$ be argumentation structures, $\shead(r_2) \in Z_1$, $X_1 \cup \sbodym{r_2} \cup Z_1 \cup \sbodyp{r_2} \cup \{ \shead(r_1) \}$ be consistent and $X_1 \cup \sbodym{r_2}$ be self-consistent.

Then also $\argStName{A_3} = \argStFull{\shead(r_1)}{X_1 \cup \sbodym{r_2}}{(Z_1 \setminus \{ \shead(r_2) \}) \cup \sbodyp{r_2}}$ is an argumentation structure. We also write $\argStName{A_3} = u(\argStName{A_1}, \argStName{A_2})$. We define $u$ as a symmetric relation: $u(\argStName{A_1}, \argStName{A_2}) = u(\argStName{A_2}, \argStName{A_1})$\footnote{Symmetry of $u$ enables below a more succinct formulation of derivation rules Q1, Q2. The symbol $u$ indicates that $\argStName{A_3}$ is a result of an unfolding.}

\item[R2]
Let $\argStA 1$ and $\argStA 2$ be argumentation structures and $X_1 \cup X_2 \cup Y_1 \cup Y_2$ be consistent and $X_1 \cup X_2$ be self-consistent. 

Then also $\argStName{A_3} = \argStSimple{Y_1 \cup Y_2}{X_1 \cup X_2}$ is an argumentation structure. We also write $\argStName{A_3} = \argStName{A_1} \cup \argStName{A_2}$.

\item[R3] Let $\argStA 1$ be an argumentation structure and $W \subseteq \sDef$. 

If $X_1 \cup W \cup Y_1$ is consistent and $X_1 \cup W$ is self-consistent, then also $\argStName{A_2} = \argStSimple{Y_1}{X_1 \cup W}$ is an argumentation structure. We also write $\argStName{A_2} = \argStName{A_1} \cup W$.$\Box$
\end{description}
\end{df}

\begin{ex}[\cite{be}]
\label{zaciatok}
Let a program $P$ be given ($P$ is used as a running example in this paper): \\
\hspace*{0.3cm} $r_1$ \hspace{2cm} $b \la a, \snot \neg b$ \\
\hspace*{0.3cm} $r_2$ \hspace{2cm} $\neg b \la \snot b$ \\
\hspace*{0.3cm} $r_3$ \hspace{2cm} $a \la \snot \neg a$.

\noindent
Suppose that $\prec = \{ (r_2, r_1) \}$.

Consider the rule $r_2$. The default negation $\snot b$ plays the role of  a {\em defeasible argument}. If the argument can be consistently evaluated as true with respect to a program containing $r_2$, then also $\neg b$ can (and must) be evaluated as true.

As regards the rule $r_1$, default negation $\snot \neg b$ can be treated as an argument for $b$, if $a$ is true, it is an example of a ``conditional argument''.

The following basic argumentation structures correspond to the rules of $P$:
\\
$\langle \{ b \} \hookleftarrow \{ \snot \neg b \} ; \{ a \} \rangle$,$\langle \{ \neg b \} \hookleftarrow  \{ \snot b \} \rangle$, 
$\langle \{ a \} \hookleftarrow \{ \snot \neg a \} \rangle$.
Let us denote them by $\mathcal A_1, \mathcal A_2, \mathcal A_3$, respectively. 

Some arguments can be treated as counterarguments against other arguments. If  we accept the argument $\snot b$ (with the consequence $\neg b$), it can be treated as a counterargument to $\snot \neg b$ and, similarly, $\snot \neg b$ (with the consequence $b$, if $a$ is true) as a counterargument against $\snot b$.  
On the level of argumentation structures it can be said that $\mathcal{A}_1$ contradicts $\mathcal{A}_2$ and vice versa.

The preference relation can be directly transferred to {\em basic} argumentation structures, hence $\mathcal{A}_1$ is more preferred than $\mathcal{A}_2$. Consequently, only the attack of $\mathcal{A}_1$ against $\mathcal{A}_2$ is relevant.

An example of a derived argumentation structure: $\mathcal{A}_3$ enables to ``unfold'' the condition $a$ in $\mathcal{A}_1$, the resulting argumentation structure can be expressed as $\mathcal{A}_4 = \langle \{ b \} \hookleftarrow \{ \snot \neg b, \snot \neg a \} \rangle$. Similarly, $\mathcal{A}_5 = \langle \{ a, b \} \hookleftarrow \{ \snot \neg b, \snot \neg a \} \rangle$ can be derived from $\mathcal{A}_3$ and $\mathcal{A}_4$, $\mathcal{A}_5 = \mathcal{A}_3 \cup \mathcal{A}_4$.

We will also transfer the attack relation from the basic argumentation structures to the derived argumentation structures. 

Observe that some argumentation structures correspond to answer sets. $\mathcal{A}_5$ corresponds to the answer set $\{ a, b \}$ and $\mathcal{A}_6 = \langle \{ a, \neg b \} \hookleftarrow \{ \snot b, \snot \neg a \} \rangle$
to $\{ a, \neg b \}$. Notice that $\mathcal{A}_6 = \mathcal{A}_2 \cup \mathcal{A}_3$. The attack relation enables to select the preferred answer set. This will be discussed later in Example \ref{preverit}.
$\Box$
\end{ex}

\begin{pr}
Each argumentation structure is a dependency structure.
\end{pr}
\begin{proof}
We have to show that an application of R1, R2 and R3 preserves properties of dependency structures.
\begin{description}
\item[R1]
Since $S_1 = X_1 \cup \sbodym{r_2} \cup Z_1 \cup \sbodyp{r_2} \cup \{ \shead(r_1) \}$ is consistent $S_2 = X_1 \cup \sbodym{r_2} \cup ( Z_1 \setminus \{\shead(r_2)\} ) \cup \sbodyp{r_2} \subseteq S_1$ is also consistent. It means $\sP(X_1 \cup \sbodym{r_2}) \cap ( 
(Z_1 \setminus \sset{\shead(r_2)} ) \cup \sbodyp{r_2}) = \emptyset$.

Let $Q = P \cup (Z_1 \setminus \sset{\shead(r_2)}) \cup \sbodyp{r_2}$ and $w = \shead(r_2) \la$.

From $\shead(r_2) \in \sCn{P \cup \sbodyp{r_2}}{\sbodym{r_2}}$ we have sequence of rules $R_2 = \langle q_1, q_2, \dots, q_m \rangle$ where $m >0$ and $q_m = r_2$. 

From $\shead(r_1) \in \sCn{P \cup Z_1}{X_1}$ we have sequence of rules $R_1 = \langle p_1, p_2, \dots, p_n \rangle$ where $n >0$ and $p_n = r_1$. We assume there is at most one occurrence of $w$ in $R_1$. Otherwise we can remove all but leftmost one. Note that since $r_2 \in P$ there is a possibility to satisfy $\sbody(r_1)$ in a different way than using $w$.

If $w \in R_1$ then we have $p_i = w$ for some $1 \leq i < n$. We construct sequence
\\ 
$R_3 = \langle q_1, q_2, \dots, q_m, p_1, p_2, \dots, p_{i-1}, p_{i+1}, \dots, p_n \rangle$. If $w \not \in R_1$ we construct sequence
\\
$R_3 = \langle q_1, q_2, \dots, q_m, p_1, p_2, \dots, p_n \rangle$. In both cases $R_3$ satisfy conditions from definition \ref{depP} for assumption $X_1 \cup \sbodym{r_2}$.

Since rules in R3 are from program $Q$ we have $\shead(r_1) \in \sCn{Q}{X_1 \cup \sbodym{r_2}}$.

\item[R3]
$Z_2 = \emptyset$ hence $\sP(X_1 \cup W) \cap Z_2= \emptyset$. We have $Y_1 \subseteq \sCn{P}{X_1}$. So for every $y \in Y_1$ there is a sequence $R$ of rules that satisfy Definition \ref{depP} for assumption $X_1$. Same sequence satisfy definition \ref{depP} for superset assumption $X_1 \cup W$. Hence $y \in \sCn{P}{X_1 \cup W}$ and $Y_1 \subseteq \sCn{P}{X_1 \cup W}$.

\item[R2]

$Z_3 = \emptyset$ hence $\sP(X_1 \cup X_2) \cap Z_3 = \emptyset$. We have $Y_1 \subseteq \sCn{P}{X_1}$ hence $Y_1 \subseteq \sCn{P}{X_1 \cup X_2}$. We also have $Y_2 \subseteq \sCn{P}{X_2}$ hence $Y_2 \subseteq \sCn{P}{X_1 \cup X_2}$. Therefore $Y_1 \cup Y_2 \subseteq \sCn{P}{X_1 \cup X_2}$.
\end{description}
\end{proof}

A {\em derivation} of an argumentation structure $\mathcal A$ (w.r.t. $P$) is a sequence $\langle \sM 1, \sM 2, \dots, \sM k \rangle$ of argumentation structures (w.r.t. $P$) such that $\sM 1$ is a basic argumentation structure, $\mathcal A = \sM k$ and each $\sM i$, $1 < i \leq k$, is either a basic argumentation structure or it is obtained by R1 or R2 or R3 from preceding argumentation structures.

\section{Attacks}
\label{satt}

Our approach to preferred answer sets is based on a solution of conflicts between argumentation structures. We distinguish three steps towards that goal. {\em Contradictions} between argumentation structures represent the elementary step. Rule preference and contradiction between basic argumentation structures are used to form an attack relation. Consider two basic argumentation structures $\sM 1$ and $\sM 2$. If $\sM 1$ contradicts $\sM 2$ and corresponds to a more preferred rule, then it {\em attacks} $\sM 2$. Attacks are propagated to other argumentation structures using derivation rules.  Attacks between argumentation structures depend on how argumentation structures are derived. Hence, we need a more context-independent notion and we define a relation of {\em blocking} between argumentation structures. The complement of blocking (warranting) is used in the definition of preferred argumentation structures.

\begin{df}
Consider argumentation structures $\mathcal A = \langle Y_1 \hookleftarrow X_1; Z_1 \rangle$ and $\mathcal B = \langle Y_2 \hookleftarrow X_2; Z_2 \rangle$.

If there is a literal $L \in Y_1$ such that $\snot L \in X_2$, it is said that the argument $X_1$ {\em contradicts} the argument $X_2$ and the argumentation structure $\mathcal A$ {\em contradicts} the argumentation structure $\mathcal B$.

It is said that $X_1$ is a {\em counterargument} to $X_2$.
$\Box$
\end{df}

The basic argumentation structures corresponding to the facts of the given program are not contradicted.

Let $r_1 = a \la$ be a fact and $\snot a \in \sbodym{r_2}$.
Then any $W \subseteq \sDef$ s.t. $\sbodym{r_2} \subseteq W$ is not self-consistent and, therefore, it is not an argument.

\begin{ex}
\label{shows}
In Example \ref{zaciatok}, $\mathcal A_1$ contradicts $\mathcal A_2$ and $\mathcal A_2$ contradicts $\mathcal A_1$.

Only some counterarguments are interesting: the rule $r_1$ is more preferred than the rule $r_2$, therefore the counterargument of $\sM 2$ against $\sM 1$ should not be ``effectual''. We are going to introduce a notion of {\em attack} in order to denote ``effectual'' counterarguments.
$\Box$
\end{ex}

Similarly as for the case of argumentation structures, the basic attacks are defined first. A terminological convention: if $\sM 1$ attacks $\sM 2$, it is said that the pair $(\sM 1, \sM 2)$ is an attack.

\begin{df}
Let $r_2 \prec r_1$ and let $\argStC 1$ contradicts $\argStC 2$. 

Then $\sM 1$ {\em attacks} $\sM 2$ and it is said that this attack is {\em basic}. 
$\Box$
\end{df}

Attacks between argumentation structures ``inherited'' (propagated) from basic attacks are defined in terms of derivation rules. The rules of that inheritance are motivated and defined below.

\begin{ex}
Let us continue with Example \ref{zaciatok}.

Consider the basic argumentation structures
$\mathcal A_1 = \langle \{ b \} \hookleftarrow \{ \snot \neg b \} ; \{ a \} \rangle$, 
$\mathcal A_2 = \langle \{ \neg b \} \hookleftarrow  \{ \snot b \} \rangle$, 
$\mathcal A_3 = \langle \{ a \} \hookleftarrow \{ \snot \neg a \} \rangle$
and the derived argumentation structures $\mathcal A_4 = \langle \{ b \} \hookleftarrow \{ \snot \neg b, \snot \neg a \} \rangle$, $\mathcal A_5 = \langle \{ b, a \} \hookleftarrow \{ \snot \neg b, \snot \neg a \} \rangle$, $\mathcal A_6 = \langle \{ \neg b, a \} \hookleftarrow \{ \snot b, \snot \neg a \} \rangle$.

$(\sM 1, \sM 2)$ is the only basic attack (the more preferred $\sM 1$ attacks the less preferred $\sM 2$).

Derivations of the attacks of $(\mathcal A_4, \sM 2)$ and $(\mathcal A_5, \sM 2)$ could be motivated as follows. 
$\sM 4$ is derived from $\sM 1$ and $\sM 3$ using R1, the attack of $\sM 1$ against $\sM 2$ should be propagated to the attack $(\sM 4, \sM 2)$. Note that $\sM 3$ is not attacked.

Now, $\sM 5$ is derived from $\sM 3$ and $\sM 4$. Again, the attack of $\sM 4$ against $\sM 2$ should be inherited by $(\sM 5, \sM 2)$.

Similarly, $\mathcal{A}_6$ is derived from attacked $\mathcal{A}_2$. The attacks against $\mathcal{A}_2$ are transferred to the attacks against $\mathcal{A}_6$. The attack $(\sM 5, \sM 6)$ is a crucial one, a selection of preferred answer set is based on it; compare with Example \ref{preverit}.

On the contrary, $\sM 2$ contradicts $\sM 4$ and $\sM 5$, but it is based on a less preferred rule, hence those contradictions are not considered as attacks.
$\Box$
\end{ex}

First we define two rules, Q1 and Q2, which specify inheritance of attacks ``via unfolding'' - use of the rule R1. Second, two rules Q3 and Q4 derive attacks when the  attacking or attacked side is joined with another argumentation structure - use of the rule R2. Finally, rules Q5 and Q6 derive attacks, if attacking or attacked side is extended by an assumption - use of the rule R3. Some asymmetries between pairs Q1, Q2 and Q3, Q4 will be discussed below, see Example \ref{troubles}.

\begin{df}[Attack derivation rules]
Basic attacks are attacks.
\begin{description}

\item[Q1]
Let $\argStName{A_1}, \argStName{A_2},\argStName{A_3}$ be argumentation structures such that:
\begin{itemize}
\item $\argStName{A_1}$ attacks $\argStName{A_2}$,
\item $\argStName{A_3}$ does not attack $\argStName{A_1}$, and
\item $u(\argStName{A_2}, \argStName{A_3})$ is argumentation structure.
\end{itemize}
Then $\argStName{A_1}$ attacks $u(\argStName{A_2}, \argStName{A_3})$.

\item[Q2]
Let $\argStName{A_1}, \argStName{A_2}, \argStName{A_3}$ be argumentation structures such that:
\begin{itemize}
\item $\argStName{A_1}$ attacks $\argStName{A_2}$,
\item $\argStName{A_3}$ is not attacked, and
\item $u(\argStName{A_1}, \argStName{A_3})$ is argumentation structure.
\end{itemize}
Then $u(\argStName{A_1}, \argStName{A_3})$ attacks $\argStName{A_2}$.

\item[Q3]
Let $\argStName{A_1}$  and $\argStName{A_3}$ be argumentation structures of the form $\argStSimple{X}{Y}$ and $\argStName{A_2}$ be an argumentation structure. Suppose that:
\begin{itemize}
\item $\argStName{A_1}$ attacks $\argStName{A_2}$, 
\item $\argStName{A_3}$ is not attacked and
\item $\argStName{A_1} \cup \argStName{A_3}$ is argumentation structure.
\end{itemize}
Then $\argStName{A_1} \cup \argStName{A_3}$ attacks $\argStName{A_2}$.

\item[Q4]
Let $\argStName{A_1}$ be an argumentation structure and $\argStName{A_2}$,  $\argStName{A_3}$ be argumentation structures of the form $\argStSimple{X}{Y}$ such that:
\begin{itemize}
\item $\argStName{A_1}$ attacks $\argStName{A_2}$,
\item $\argStName{A_3}$ does not attack $\argStName{A_1}$, and
\item $\argStName{A_2} \cup \argStName{A_3}$ is argumentation structure.
\end{itemize}
Then $\argStName{A_1}$ attacks $\argStName{A_2} \cup \argStName{A_3}$.

\item[Q5]
Let $\argStName{A_1} = \argStSimple{X}{Y}$ and $\argStName{A_2}$ be argumentation structures. Let $W \subseteq \sDef$. Suppose that:
\begin{itemize}
\item $\argStName{A_1}$ attacks $\argStName{A_2}$, and
\item $\argStName{A_1} \cup W = \argStSimple{X}{Y \cup W}$ is argumentation structure.
\end{itemize}
Then $\argStName{A_1} \cup W$ attacks $\argStName{A_2}$.

\item[Q6]
Let $\argStName{A_1}$ and $\argStName{A_2} = \argStSimple{X}{Y}$ be argumentation structures. Let $W \subseteq \sDef$. Suppose that:
\begin{itemize}
\item $\argStName{A_1}$ attacks $\argStName{A_2}$, and
\item $\argStSimple{X}{Y \cup W} = \argStName{A_2} \cup W$ is argumentation structure.
\end{itemize}
Then $\argStName{A_1}$ attacks  $\argStName{A_2} \cup W$.
\end{description}
There are no other attacks except those specified above.
$\Box$
\end{df}

\begin{df}
A {\em derivation of an attack} $\mathcal X$ is a sequence $\mathcal X_1, \dots, \mathcal X_k$, where $\mathcal X = \mathcal X_k$, each $\mathcal X_i$ is an attack (a pair of attacking and attacked argumentation structures), $\mathcal X_1$ is a basic attack and each $\mathcal X_i$, $1 < i \leq k$ is either a basic attack or it is derived from the previous attacks using rules Q1, Q2, Q3, Q4, Q5, Q6.
\end{df}

Derivations of argumentation structures and of attacks blend together. Example \ref{ambiguity} shows that a pair of argumentation structures $(\mathcal B, \mathcal A)$ is an attack w.r.t. a derivation, but it is not an attack w.r.t another derivation. Let us start with a definition.

\begin{df}
Let a program $P$ and an answer set $S$ be given. Let be
$R = \{ r \in P \mid \sbody(r) \subseteq S \}$.
It is said that $R$ is the set of all {\em generating rules} of $S^{+}$.
$\Box$
\end{df} 

\begin{ex}
\label{ambiguity}
Let $P$ be \\
\hspace*{0.3cm} $r_1$ \hspace{2cm} $a \la \snot b$   \\
\hspace*{0.3cm} $r_2$ \hspace{2cm} $b \la \snot a$  \\
\hspace*{0.3cm} $r_3$ \hspace{2cm} $a \la \snot c$ \\
\hspace*{0.3cm} $r_4$ \hspace{2cm} $c \la  b$. 

\noindent
$\prec = \{ (r_1, r_2) \}$. 

There are two answer sets of $P$: $S_1 = \{ a \}$ and $S_2 = \{ b, c \}$. The corresponding argumentation structures are $\mathcal{A} = \langle \{ a \} \hookleftarrow \{ \snot b, \snot c \} \rangle$ and  $\mathcal{B} = \langle \{ b, c \} \hookleftarrow \{ \snot a \} \rangle$, respectively. 

There are two derivations of $\mathcal{A}$. Both derivations start from a basic argumentation structure and R3 is used. The first is the sequence $\mathcal{A}_1, \mathcal A$ and the second is $\mathcal{A}_3, \mathcal A$, where $\mathcal{A}_1 = \langle \{ a \} \hookleftarrow \{ \snot b \} \rangle$ and $\mathcal{A}_3 = \langle \{ a \} \hookleftarrow \{ \snot c \} \rangle$. 

If the sequence $\mathcal{A}_1, \mathcal A$ is considered, 
an attack against $\mathcal A$ is derivable. 
Let be $\mathcal{A}_2 = \langle \{ b \} \hookleftarrow \{ \snot a \} \rangle$, $\mathcal{A}_4 = \langle \{ c \} \hookleftarrow \emptyset; \{ b \} \} \rangle$. The corresponding attack derivation is as
follows: \\ 
$(\mathcal{A}_2, \mathcal{A}_1), (u(\mathcal{A}_4, \mathcal{A}_2), \mathcal{A}_1), (\mathcal{B}, \mathcal{A}_1), (\mathcal{B}, \mathcal{A})$, where Q2, Q3 and Q6 are used. 

The only basic attack of our example is $(\mathcal{A}_2, \mathcal{A}_1)$. 
Hence, the second derivation $\mathcal{A}_3, \mathcal A$ of $\mathcal{A}$ cannot be attacked.

The derivations of $\mathcal{A}$ correspond to two sets of rules generating $S_1$, i.e., $R_1 = \{ r_1 \}$, and $R_2 = \{ r_3 \}$.
$R_1$ contains an attacked rule, while $R_2$ does not contain such a rule. We accept that if there is a set of rules generating an answer set $S$ s.t. no rule is attacked by a rule generating another answer set, then we can consider $S$ as a preferred one.

We transfer this rather credulous approach to derivations of preferred argumentation structures.
$\Box$
\end{ex}

\begin{df}[Complete arguments]

An argumentation structure $\langle Y\hookleftarrow X \rangle$ is called {\em complete} iff for each literal $L \in \sObj$ it holds that $L  \in Y$ or $\snot L \in X$.
$\Box$
\end{df}

\begin{df}[Warranted and blocked derivations]
\label{warBl}
Let $\sigma = \mathcal{A}_1, \dots, \mathcal{A}_k$ be a derivation of an argumentation structure $\mathcal{A}$, where $\mathcal{A} = \mathcal{A}_k$.

It is said that $\sigma$ is {\em blocked} iff 
there is a derivation $\tau$ of the attack $(\mathcal B, \mathcal{A})$, where $\mathcal B$ is a complete argumentation structure and each member of $\tau$ contains an $\mathcal{A}_i$ as a second component. 

A derivation is {\em warranted} if it is not blocked.
$\Box$
\end{df}

\begin{df}[Warranted and blocked argumentation structures]
An argumentation structure $\mathcal A$ is warranted iff there is a warranted derivation of $\mathcal A$.

$\mathcal A$ is blocked iff each derivation of $\mathcal A$ is blocked.
$\Box$
\end{df}

\section{Preferred answer sets}

\begin{ex}
\label{preverit}
Consider our running example, where we have complete argumentation structures
$\mathcal A_5 = \langle \{ b, a \} \hookleftarrow \{ \snot \neg b, \snot \neg a \} \rangle, \mathcal A_6 = \langle \{ \neg b, a \}\hookleftarrow \{ \snot \neg a, \snot b \} \rangle$.

We will prefer $\mathcal A_5$ over $\mathcal A_6$. $\mathcal A_6$ is blocked by $\mathcal{A}_5$. On the other hand, $\mathcal A_5$ is not blocked.  

Consequently, we will consider $\{ a, b \}$ as a preferred answer set of the given prioritized logic program. 
$\Box$
\end{ex}

\begin{df}[Preferred answer set]
An argumentation structure is {\em preferred} iff it is complete and warranted.

$Y \cup X$ is a {\em preferred answer set} iff $ \langle Y \hookleftarrow X \rangle$ is a preferred argumentation structure.
$\Box$
\end{df}

Notice that our notion of preferred answer set is rather a credulous one, it is based on the notion of warranted derivation, i.e., at least one derivation of a preferred answer set should not be blocked.

The following example shows that the argumentation structure corresponding to the only answer set of a program is preferred, even if it is attacked (by an argumentation structure which is not complete).

\begin{ex}
\begin{eqnarray*}
r_1 \hspace{2cm} b & \la & \snot a \\
r_2 \hspace{2cm} a & \la & \snot b \\
r_3 \hspace{2cm} c & \la & a \\
r_4 \hspace{2cm} c & \la & \snot c
\end{eqnarray*}
$\prec = \{ (r_1, r_2), (r_3, r_4) \}$.

Let the basic argumentation structures be denoted by $\sM i$, $i =1, \dots, 4$. $(\sM 1, \sM 2), (\sM 3, \sM 4)$ are the basic attacks. $\sM 1$ attacks $\sM 5 = \langle \{ c \} \hookleftarrow \{ \snot b \} \rangle$ according to the rule Q1 and  $\sM 1$ attacks $\sM 6 = \langle \{ c, a \} \hookleftarrow \{ \snot b \} \rangle$ according to the rule Q4.

Remind that according to Definition \ref{warBl} a derivation can be blocked only by a complete argumentation structure and an argumentation structure is blocked iff each its derivation is blocked. Consequently, the complete argumentation structure $\sM 6$ is not blocked by another complete argumentation structure (there is no such structure) and, consequently, it is the preferred argumentation structure.

We distinguish between attacking and blocking. A blocked argumentation structure is attacked by a {\em complete} argumentation structure. Preferred argumentation structures are not blocked.
$\Box$
\end{ex}

Next example explains asymmetries between Q1, Q2 and Q3, Q4. The main idea is as follows. We are more cautious when an attacking argumentation structure is derived (Q2, Q3) and we require that a ``parent'' of the attacking argumentation structure is not attacked at all. On the other hand, a scheme of derivation rules Q1 and Q4 is as follows: $\mathcal{A}_1$ attacks $\mathcal{A}_2$, $\mathcal{A}$ is a derived argumentation structure from the attacked $\mathcal{A}_2$ and an argumentation structure $\mathcal{A}_3$. In order to derive an attack of $\mathcal{A}_1$ against $\mathcal{A}$ it is sufficient to assume that $\mathcal{A}_3$ does not attack $\mathcal{A}_1$. However, there are some problems with this design decision.

\begin{ex}
\label{troubles}
Consider a program $P$:\\
\hspace*{0.5cm} $r_1$ \hspace{1cm} $a_1 \la \snot a_3, \snot d_2$ \\
\hspace*{0.5cm} $r_2$ \hspace{1cm} $d_1 \la \snot a_3, \snot d_2$ \\
\hspace*{0.5cm} $r_3$ \hspace{1cm} $a_2 \la \snot a_1, \snot d_3$ \\
\hspace*{0.5cm} $r_4$ \hspace{1cm} $d_2 \la \snot a_1, \snot d_3$ \\
\hspace*{0.5cm} $r_5$ \hspace{1cm} $a_3 \la \snot a_2, \snot d_1$ \\
\hspace*{0.5cm} $r_6$ \hspace{1cm} $d_3 \la \snot a_2, \snot d_1$

\noindent
$\prec = \{ (r1, r4), (r3, r5), (r6, r2) \}$.

Notice that a complete argumentation structure $\mathcal{B}_1 = \argStSimple{\{a_1, d_1\}}{\{\snot a_3, \snot d_2\}}$ is derived from $\mathcal{A}_1$ corresponding to $r_1$ and from $\mathcal{A}_2$ corresponding to $r_2$, similarly $\mathcal{B}_2 = \argStSimple{\{a_2, d_2\}}{\{\snot a_1, \snot d_3\}}$ is derived from $\mathcal{A}_3$ corresponding to $r_3$ and from $\mathcal{A}_4$ corresponding to $r_4$ and $\mathcal{B}_3 = \argStSimple{\{a_3, d_3\}}{\{\snot a_2, \snot d_1\}}$ is derived from $\mathcal{A}_5$ corresponding to $r_5$ and from $\mathcal{A}_6$ corresponding to $r_6$. $\mathcal{B}_1, \mathcal{B}_2, \mathcal{B}_3$ are all complete argumentation structures of our example.

Suppose that Q3 does not contain condition that a ``parent'' of the attacking argumentation structure is not attacked at all. Then we get that each complete argumentation structure is blocked, consequently, no preferred answer set can be selected. But we are extremely interested in a selection of a preferred answer set.

As a consequence, we are too liberal in selecting preferred answer sets: Consider program \\
\hspace*{0.5cm} $r_1$ \hspace{1cm} $a \la $ \\
\hspace*{0.5cm} $r_2$ \hspace{1cm} $b \la \snot a$ \\
\hspace*{0.5cm} $r_3$ \hspace{1cm} $c \la \snot b$ \\
\hspace*{0.5cm} $r_4$ \hspace{1cm} $b \la \snot c$ 

\noindent
$\prec= \{ (r2, r3), (r3, r4) \}$.

We get that both $S_1 = \{ a, c \}$ and $S_2 = \{ a, b \}$ are preferred answer sets,
but $S_2$ is not an intuitive selection. The reason is that both argumentation structures (let us denote them by $\mathcal{A}_1$ and $\mathcal{A}_2$) corresponding to $S_1$ and $S_2$, respectively, are attacked and rule Q3 cannot be applied. Hence, each derivation of $\mathcal{A}_1$ and $\mathcal{A}_2$ is warranted. A weaker version of Q3 would enable to repair this, however, it is a too high price for us.
$\Box$
\end{ex}

\begin{te}
If $S$ is a preferred answer set of $(P, \prec, \mathcal N)$, then $S$ is an answer set of $P$.
\end{te}

\section{Principles}

The principles (partially) specify what it means that an order on answer sets corresponds to the given order on rules. The first two principles below are dependent on \cite{be}.
Principle III reproduces an idea of Proposition 6.1 from \cite{be}. 
The Principles of \cite{be} are originally expressed in an abstract way for the general case of nonmonotonic prioritized defeasible rules. We restrict the discussion (and the wording) of the Principles
to the case of logic programs and answer sets.

Let $P$ be a program and $r_1, r_2 \in P$. It is said that $r_2$ is {\em attacked} by $r_1$ ($r_1$ attacks $r_2$) iff $\snot \shead(r_1) \in \sbodym{r_2}$.

\begin{df}
\label{atGenR}
Let a prioritized logic program $(P, \prec, \mathcal N)$ be given. Let $R_1, R_2$ be sets of generating rules for some answer sets of $P$. It is said that $R_1$ {\em attacks} $R_2$ iff there is $r_1 \in R_1, r_2 \in R_2$ such that $r_2 \prec r_1$ and $r_1$ attacks $r_2$.
\end{df}

\begin{df}
\label{warGenR}
Let a prioritized logic program $(P, \prec, \mathcal N)$ be given. Let $R$ be a set of generating rules of some answer set of $P$. It is said that $R$ is a {\em warranted set of generating rules} iff there is no set $Q$ of generating rules of some answer set of $P$ such that $Q$ attacks $R$ .
$\Box$
\end{df}

Principle I in its original formulation does not hold for our approach. A terminological remark -- words associated to our approach (attack, warranted) are used in presented formulations of Principles. But remind definitions \ref{atGenR} and \ref{warGenR} -- the notions are defined for generating sets of rules independently on our approach. We have considered the same principles also in another approach, see \cite{sano} and also \cite{dc-iclp}. It is defined directly on generating sets, and uses neither argumentation structures nor derivation rules.

In all principles below it is assumed that a prioritized logic program $(P, \prec, \mathcal N)$ is given.\footnote{The original formulation of principles by \cite{be} is as follows.

Principle I. Let $B_1$ and $B_2$ be two belief sets of a prioritized theory $(T; \prec)$ generated by the (ground) rules $R \cup d_1$ and $R \cup d_2$, where $d_1, d_2 \not \in R$, respectively. If $d_1$ is preferred over $d_2$, then $B_2$ is not                             a (maximally) preferred belief set of $T$.

Principle II. Let $B$ be a preferred belief set of a prioritized theory $(T; \prec)$ and $r$ a (ground) rule such that at least one prerequisite of $r$ is not in $B$. Then $B$ is a preferred belief set of $(T \cup \{ r \}; \prec^{\prime} )$ whenever $\prec^{\prime}$
      agrees with $\prec$ on priorities among rules in $T$.
}

{\bf Principle I} Let $A_1$ and $A_2$ be two answer sets of the program $P$. Let $R \subset P$ be a set of rules and $d_1, d_2 \in P \setminus R$ are rules. Let $A_{1}^{+}, A_{2}^{+}$ be generated by the rules $R \cup \{ d_1 \}$ and $R \cup \{ d_2 \}$, respectively. If $d_1$ is preferred over $d_2$ and each set of generating rules of $\mathcal{A}_2^{+}$ is attacked by a warranted set of generating rules of some answer set of $P$, then $A_2$ is not a preferred answer set of $(P, \prec, \mathcal N)$.
$\Box$

Our formulation of Principle I differs from the original formulation in \cite{be} -- the condition ``each set of generating rules of $\mathcal{A}_2^{+}$ is attacked by a warranted set of generating rules of some answer set of $P$'' is added because of the credulous stance to warranted derivations accepted in this paper

We do not accept the following principle. See discussion below.

{\bf Principle II}
Let $A$ be a preferred answer set of a prioritized logic program $(P, \prec, \mathcal N)$ and $r$ be a rule such that 
$\sbodyp{r} \not \subseteq A^{+}$. 
Then $A$ is a preferred answer set of $(P \cup \{ r \}, \prec^{\prime}, \mathcal N^{\prime})$, whenever $\prec^{\prime}$ agrees with $\prec$ on rules in $P$ and $\mathcal N^{\prime}$ extends $\mathcal N$ with the name $r$.
$\Box$

We believe that the possibility to always select  a preferred answer set from a non-empty set of standard answer sets is of critical importance. 

{\bf Principle III}
Let $\mathcal B \not = \emptyset$ be the set of all answer sets 
of $P$.
Then there is a selection function $\Sigma$ s.t. $\Sigma(\mathcal B)$ is the set of  all preferred answer sets of $(P, \prec, \mathcal N)$, where $\emptyset \not = \Sigma(\mathcal B) \subseteq \mathcal B$.
$\Box$

It is shown in \cite{be}, Proposition 6.1, that Principle II is incompatible with  Principle III, if the notion of preferred answer set from \cite{be} is accepted:

\begin{ex}[\cite{be}]
\label{tamtonieje}
Consider program $P$, whose single standard answer set is $S = \{ b \}$ and the rule (\ref{oone}) is preferred over the rule (\ref{ttwo}).
\begin{eqnarray} 
c & \la & \snot b \label{oone} \\
b & \la & \snot a \label{ttwo}
\end{eqnarray}
$S$ is not a preferred answer set in the framework of \cite{be}.
Assume that $S$, the only standard answer set of $P$, is selected -- according
to the Principle III  -- as the preferred answer set of $(P, \prec, \mathcal N)$.\footnote{Observe that the only derived complete argumentation structure is $\langle \{ b \}\hookleftarrow \{ \snot a, \snot c \} \rangle$. Hence, $\{ b \}$ is a preferred answer set of $(P, \prec,\mathcal N)$ in our framework.}
Let $P^{\prime}$ be $P \cup \{ a \la c \}$ and $a \la c$ be preferred over the both rules \ref{oone} and \ref{ttwo}.
$P^{\prime}$ has two standard answer sets, $S$ and $T = \{ a,c \}$.
Note that $\{ c \} \not \subseteq S^{+}$.  Hence, $S$ should be the preferred answer set of $P^{\prime}$ according to the Principle II. However, in the framework  of \cite{be} the only preferred answer set of $(P^{\prime}, \prec^{\prime}, \mathcal N^{\prime})$ is $T$. This selection of preferred answer set satisfies clear intuitions -- $T$ is generated by the two most preferred rules. 

In our approach the complete argumentation structure $\sM 4 = \langle \{ a, c \} \hookleftarrow \{ \snot b \} \rangle$ blocks the complete argumentation structure $\sM 5 = \langle \{ b \} \hookleftarrow \{ \snot a, \snot c \} \rangle$, hence, $\sM 4$ is preferred and $\{a, c \}$ is the preferred answer set.

Principle III is of crucial value according to our view. 
A more detailed justification of our decision not to accept Principle II is presented in \cite{nmr08}. We mention here only that  
we select preferred answer sets of $P^{\prime}$ from a broader variety of possibilities.
Consequently, no condition satisfied by a subset of those possibilities should constrain the selection of  preferred answer set from the extended set of possibilities.
$\Box$
\end{ex}
Principle II is not accepted also in \cite{si}. According to \cite{class} descriptive approaches do not satisfy this principle in general.




Principle IV expresses when an answer set is a preferred one.
We consider it as an expression of a descriptive approach to preferred answer set specification, as we understand it and accept in this paper.

{\bf Principle IV}
Let $S$ be answer set of $P$. Suppose that $S$ is generated by a set of rules $R$. If $R$ is a warranted set of generating rules then $S$ is a preferred answer set. $\Box$

As regards a choice of principles, we accept the position of \cite{be}: even if somebody does not accept a set of principles for preferential reasoning, those (and similar) principles are still of interest as they may be used for classifying different patterns of reasoning.

\section{Discussion}

The question whether derivation rules for attacks are sufficient and necessary arises in a natural way. Our only response to the question in this paper is that Principles I, III, IV are satisfied, when we use notions of attack, blocking and warranting introduced in this paper
We proceed to theorems about satisfaction of principles. 

\begin{te}
Principle III is satisfied. Let $\mathcal P = (P, \prec, \mathcal N)$ be a prioritized logic program and $\sSM(P) \not = \emptyset$. 
Then there is a preferred answer set of $\mathcal P$.
\end{te}
\begin{proof}
{\bf Case 1} is trivial -- if a program $P$ have only one answer set $S$, then no complete argumentation structure blocks $\argStSimple{S^{+}}{S^{-}}$.

{\bf Case 2}. Let a program $P$ has only two answer sets $S_1$ and $S_2$. Let the corresponding complete argumentation structures be $\argStName{A_1} = \argStSimple{S_1^{+}}{S_1^{-}}$ and $\argStName{A_2} = \argStSimple{S_2^{+}}{S_2^{-}}$, respectively. Suppose that $\argStName{A_1}$ and $\argStName{A_2}$ block each other. 

It means that each derivation of both is blocked by the other complete argumentation structure. Consider all derivations of $\argStName{A_1}$ (which should be blocked by $\argStName{A_2}$). Hence, each derivation $\sigma_i$ contains an argumentation structure $\mathcal{B}_i$ attacked by $\argStName{A_2}$, i.e., $\mathcal{X} = (\argStName{A_2}, \mathcal{B}_i)$ is an attack. Each derivation of $\mathcal{X}$ should start from a basic attack and ends with $(\argStName{A_2}$, $\argStName{B_i})$.

If $\mathcal{X}$ is a basic attack, 
then the only generating set of rules of $S_2$ contains only one rule $r = S_2^{+} \la S_2^{-}$, where $S^{+} = \{ \shead(r) \}$. We assume that there is a rule $r_1$ s.t. $r_1 \prec r$ and $\snot \shead(r) \in \sbodym{r_1}$. On the other hand, $\argStName{A_1}$ blocks $\argStName{A_2}$ and there is an $r_2 \in P$ which is among the generating rules of $S_1$, $r \prec r_2$ and $\snot \shead(r_2) \in \sbodym{r}$. 

Notice that $\argStFull{\shead(r_2)}{\sbodym{r_2}}{\sbodyp{r_2}}$ attacks $\argStName{A_2}$. If $\sbodyp{r_2} \not = \emptyset$, then a derivation of attack $(\argStName{A_2}, \argStName{A_1})$ has to use Q1 and $\argStFull{\shead(r_2)}{\sbodym{r_2}}{\sbodyp{r_2}}$. But Q1 is not applicable -- attacking argumentation structure should be not attacked. 
Similarly, if $\sbodyp{r_2} = \emptyset$, Q4 should be used, but Q4 is not
applicable because of the same reason.

Assume that $\mathcal{X}$ is not a basic attack. Then there is a basic attack as follows. \\ 
Let be $\mathcal{R}_1 = \argStFull{\shead(r_1)}{\sbodym{r_1}}{\sbodyp{r_1}}$, \\
$\mathcal{R}_2 = \argStFull{\shead(r_2)}{\sbodym{r_2}}{\sbodyp{r_2}}$, 
where \\ 
$\shead(r_2) \in S_2, \shead(r_1) \in S_1, r_1 \prec r_2, \snot \shead(r_2) \in \sbodym{r_1}$
and, consequently, $(\mathcal{R}_2, \mathcal{R}_1)$ is a basic attack.  

We will prove that if each derivation of $\mathcal{A}_2$ is blocked by $\mathcal{A}_1$, then it is impossible to derive the attack $(\mathcal{A}_2, \mathcal{A}_1)$.

Let the basic attack $(\mathcal{R}_2, \mathcal{R}_1)$ be given. A derivation of $(\mathcal{A}_2, \mathcal{A}_1)$ from the basic attack should contain rules Q2 or Q3 or Q4 or Q5 in order to proceed from $\mathcal{R}_2$ to $\mathcal{A}_2$ ($\mathcal{X}$ is not a basic attack). A derivation of $\mathcal{A}_2$ using R1, R2, R3 could be reconstructed from this. The derivation is blocked. Therefore, Q2, Q3, Q4 and Q5 are not applicable and the derivation of $(\mathcal{A}_2, \mathcal{A}_1)$ is impossible.

{\bf Case 3}. Let be $\sSM(P) = \{ S_1, \dots, S_k \}, k \geq 3$. Assume that the corresponding complete argumentation structures are $\argStName{A_i}, i = 1, \dots, k$. Suppose that each of them is blocked. Let us denote the set $\{ \argStName{A_i} \mid i = 1, \dots, k \}$ by $O$.

Suppose that the set $N \subseteq O$ contains only blocked, 
but not blocking complete argumentation structures (each $\mathcal{A} \in N$ is blocked and not blocking). If $O \setminus N$ contains only basic argumentation structures then the preference relation $\prec$ is cyclic. Let $M \subseteq O$ be the set of complete argumentation structures which block an argumentation structure and they are not basic argumentation structures.

We will show that there is $\mathcal{A} \in M$ which is not blocked. 

We assumed to the contrary that each complete argumentation structure in $M$ is blocked (and blocking simultaneously). If the cardinality of $M$ is 2, Case 2 applies.

Let $\argStName{A_1}$ be in $M$, i.e., $\argStName{A_1}$ is a not basic argumentation structure. Assume (without loss of generality) that each derivation of $\argStName{A_1}$ is blocked and $\argStName{A_1}$ blocks a derivation of $\argStName{A_3}$.
We have to show that an attack $(\argStName{A_1}, \argStName{A_3})$ is not derivable.

Consider a derivation of the attack $(\argStName{A_1}, \argStName{A_3})$ and reconstruct the corresponding derivation of $\argStName{A_1}$. Suppose that  $\argStName{A_2}$ (again without loss of generality) blocks this derivation of $\argStName{A_1}$. 

Hence, $\argStName{A_2}$ attacks an argumentation structure $\mathcal B$ in the derivation of $\argStName{A_1}$. It follows that some argumentation structure in a derivation of $\argStName{A_2}$ attacks a basic argumentation structure in the derivation of $\argStName{A_1}$.
Consequently, neither rules Q1 and Q4, nor rules Q2 and Q3 are applicable in a derivation of the attack $(\argStName{A_1}, \argStName{A_3})$. Therefore, it is not derivable.
\end{proof}

Let $R \in \mathcal{R}$ be attacked by a warranted set $Q$ of generating rules for some answer set of $P$. Since $Q$ is warranted, there is a warranted derivation of complete argumentation structure $\argStName{B}$ corresponding to $Q$. There is also a derivation of complete argumentation structure $\argStName{A}$ corresponding to $R$. $Q$ attacks $R$, so there is a basic argumentation structure $\argStName{C}$ from the derivation of $\argStName{A}$ attacked by $\argStName{D}$ from derivation of $\argStName{B}$. $Q$ is warranted, rules $Q2$ and $Q3$ are applicable and hence attack $(\argStName{D}, \argStName{C})$ is propagated to attack $(\argStName{B}, \argStName{C})$. It follows that derivation of $\argStName{A}$ is blocked.


\begin{te}
\label{i}
Principle I is satisfied. 
Let $\mathcal P = (P, \prec, \mathcal N)$ be a prioritized logic program, $A_1$ and $A_2$ be two answer sets of $P$. Let $R \subset P$ be a set of rules and $d_1, d_2 \in P \setminus R$ are rules, $d_1$ is preferred over $d_2$. Let $A_{1}^{+}, A_{2}^{+}$ be generated by the rules $R \cup \{ d_1 \}$ and $R \cup \{ d_2 \}$, respectively. Assume that each set of generating rules of $\mathcal{A}_2^{+}$ is attacked by a warranted set of generating rules of some answer set of $P$.

Then $A_2$ is not a preferred answer set of $(P, \prec, \mathcal N)$.
\end{te}
\begin{proof}
It is assumed that each set of generating rules of $A_2$
 is attacked by a warranted set of generating rules of some answer set of P.
$A_2$ is not a preferred answer set of $(P, \prec , \mathcal{N})$.
\end{proof}

\begin{te}
\label{v}
Principle IV is satisfied. Let $\mathcal P = (P, \prec, \mathcal N)$ be a prioritized logic program and $S$ be an answer set of $P$. Suppose that $S$ is generated by a warranted set of rules $R$. 

Then $S$ is a preferred answer set.
\end{te}
\begin{proof}
Let $R$ be a set of rules generating an answer set $S$. If $R$ is a warranted set of generating rules, then there is a derivation of the argumentation structure $\langle S^{+} \hookleftarrow S^{-} \rangle$ which is warranted.
\end{proof}

\section{Conclusions}

An argumentation framework has been constructed, which enables transferring attacks of rules to attacks of argumentation structures and, consequently, to warranted complete argumentation structures. Preferred answer sets correspond to warranted complete argumentation structures.
This construction enables a selection of a preferred answer set whenever there is a non-empty set of standard answer sets of a program. 

We did not accept the second principle from \cite{be} and we needed to modify their first principle.
 On the other hand, new principles, which reflect the role of blocking, have been proposed. We stress the role of blocking -- in our approach, rules can be blocked by more preferred rules, but the rules which are not blocked are handled in a declarative style. 

Among goals for our future research are first of all a thorough analysis of properties and weaknesses of the presented approach (supported by an implementation of the derivation rules) and a detailed comparison to other approaches.

Finally, we have to mention the main differences between the preliminary version \cite{nmr08} and this paper. A more subtle set of attack derivation rules is introduced.
A new assumption in Q3 ($\mathcal{A}_3$ is not attacked) changed the set of attacked derivations and, consequently, our semantics. A new and more adequate notion of warranted and blocked argumentation structure is introduced, which is based on new concepts of warranted and blocked derivations. Consequently, the notion of preferred answer set is changed. A connection of attacks between argumentation structures to different derivations of argumentation structures was not expressed in \cite{nmr08}. More precise and appropriate formulations of Principles IV and I are presented. 








{\bf Acknowledgements:} We are grateful to anonymous referees for very valuable comments and proposals. This paper was supported by the grant 1/0689/10 of VEGA.

\bibliographystyle{abbrv}
\bibliography{preferences.bib}

\include{bibliography}

\end{document}